%% file: main.tex
\documentclass{article}
\usepackage[letterpaper, margin=1.1in]{geometry}

\usepackage[pagebackref=false]{hyperref}
\hypersetup{
    unicode=false,          
    colorlinks=true,        
    linkcolor=red,          
    citecolor=green,        
    filecolor=magenta,      
    urlcolor=cyan           
}

\usepackage{amsmath, amsthm, amssymb}
\usepackage{cite}
\usepackage{url}
\usepackage{cleveref}
\usepackage{appendix}
\usepackage{graphicx}

\usepackage{color}
\usepackage{algorithm}
\usepackage[noend]{algpseudocode}
\usepackage{changepage}
\usepackage{epstopdf}
\usepackage{paralist}
\usepackage{float}
\usepackage{pgf}
\usepackage{todonotes}
\usepackage[bottom]{footmisc}

\newcommand{\Tabs}{9:    \=11\=11\=11\=11\=11\=11\=11\kill}
\floatstyle{ruled}
\newfloat{alg}{htbp}{alg}
\floatname{alg}{Algorithm}
		{\end{tabbing}\end{algorithm}}

\newcommand{\ceil}[1]{\lceil #1 \rceil}
\newcommand{\Ftwo}{\mathbb{F}_{2}}

\newtheorem*{theoremRinner}{Theorem~\ref{\repthmref}}
\newenvironment{theoremR}[1]
  {\def\repthmref{#1}\theoremRinner (restated)}{\endtheoremRinner}

\newtheorem{theorem}{Theorem}[section]
\newtheorem{lemma}[theorem]{Lemma}

\newtheorem{proposition}[theorem]{Proposition}
\newtheorem{definition}[theorem]{Definition}

\renewcommand{\paragraph}[1]{\vspace{0.15cm}\noindent {\bf #1}}

\newcommand{\IncludePictures}[1]{}

\newcommand{\withCollision}[1]{}

\newcommand{\FullOrShort}{full}

\ifthenelse{\equal{\FullOrShort}{full}}{
	
	  \newcommand{\fullOnly}[1]{#1}
	  \newcommand{\shortOnly}[1]{}

  }{
	  \newcommand{\fullOnly}[1]{}
	  \newcommand{\shortOnly}[1]{#1}
  }

\begin{document}

\date{}

\author{Mohsen Ghaffari\\ \texttt{ghaffari@mit.edu}\\ MIT \and Bernhard Haeupler\\ \texttt{haeupler@cs.cmu.edu}\\ Microsoft Research \and Majid Khabbazian \\ \texttt{mkhabbazian@ualberta.ca}\\ University of Alberta}


\title{Randomized Broadcast in Radio Networks \\with Collision Detection}
\maketitle

\begin{abstract}
We present a randomized distributed algorithm that in radio networks with collision detection broadcasts a single message in $O(D + \log^6 n)$ rounds, with high probability. This time complexity is most interesting because of its optimal additive dependence on the network diameter $D$. It improves over the currently best known $O(D\log\frac{n}{D}\,+\,\log^2 n)$ algorithms, due to Czumaj and Rytter [FOCS 2003], and Kowalski and Pelc [PODC 2003]. These algorithms where designed for the model without collision detection and are optimal in that model. However, as explicitly stated by Peleg in his 2007 survey on broadcast in radio networks, it had remained an open question whether the bound can be improved with collision detection. 

\smallskip

We also study distributed algorithms for broadcasting $k$ messages from a single source to all nodes. This problem is a natural and important generalization of the single-message broadcast problem, but is in fact considerably more challenging and less understood. We show the following results: If the network topology is known to all nodes, then a $k$-message broadcast can be performed in $O(D + k\log n + \log^2 n)$ rounds, with high probability. If the topology is not known, but collision detection is available, then a $k$-message broadcast can be performed in $O(D + k\log n + \log^6 n)$ rounds, with high probability. The first bound is optimal and the second is optimal modulo the additive $O(\log^6 n)$ term. 
\end{abstract}

%
%


\input{Intro}

\input{Gathering-BFS}
\input{MMS}
\bibliographystyle{acm}
\bibliography{Bdata}
\end{document}

%% file: Intro.tex
\section{Introduction}
\renewcommand{\thefootnote}{\fnsymbol{footnote}}
\setcounter{footnote}{1}
\footnotetext{
The research in this paper was supported by AFOSR award No. FA9550-13-1-0042, NSF grant Nos. CCF-AF-0937274, CNS-1035199, 0939370-CCF, CCF-1217506,
and NSF-PURDUE-STC award 0939370-CCF.}
\setcounter{footnote}{0}
\renewcommand{\thefootnote}{\arabic{footnote}}
The classical information dissemination problem in radio networks is the problem of broadcasting a single message to all nodes of the network (single-message broadcast). This problem and its generalizations have received extensive attention.

A characteristic of radio networks is that multiple messages that arrive at a node simultaneously interfere (collide) with one another and none of them is received successfully. Regarding whether nodes can distinguish such a collision from complete silence, the model is usually divided into two categories of with and without collision detection. Throughout studies of problems in radio networks, it has been observed that many problems can be solved faster in the model with collision detection~\cite{SR10}. Despite this trend, it had remained unclear whether this is also the case for broadcast or not~\cite{Peleg07}. 
 We show that single-message broadcast can be indeed solved faster, in simply diameter plus poly-logarithmic time, if collision detection is available. Even though our work is theoretical, we remark that most practical radio networks can detect collisions. 

Broadcasting $k$ messages from one node to all nodes is a natural and important generalization of the single-message broadcast problem. Usually, this generalization involves new and significantly different challenges, mainly because the dissemination of different messages can interfere with each other. We show how to overcome these challenges and obtain an (almost) optimal $k$-message broadcast algorithm. 


\subsection{Model and Problem Statements} \label{subsec:model&prob}
We work in the \emph{radio network model with collision detection}~\cite{CK}: a synchronous network $G=(V,E)$ where in each round, each node either transmits a packet with $B$ bits or listens. As a standard assumption, to ensure that each packet can contain a constant number of ids, we assume that $B=\Omega(\log n)$. Each node $v$ receives a packet from its neighbors only if it listens in that round and exactly one of its neighbors is transmitting. If two or more neighbors of $v$ transmit, then $v$ only detects the collision, which is modeled as $v$ receiving a special symbol $\top$ indicating a collision. We explain that some of our results hold even in the model \emph{without collision detection}, where if two or more neighbors of $v$ transmit, then $v$ does not receive anything.

The single-message broadcast problem is defined as follows: A single \emph{source} node has a single message of length at most $\Theta(B)$ bits and the goal is to deliver this message to all nodes in the network. The $k$-message single-source broadcast problem is defined similarly, with the difference that the source has $k$ messages which need to be delivered to all other nodes. 
We focus on randomized solutions to these problems where we require that the message(s) are delivered to all nodes with high probability\footnote{We use the phrase ``high probability" to indicate a probability at least $1- \frac{1}{n^c}$, for a constant $c\geq 1$, and where $n$ is the network size.}. In the unknown topology setting (which is our default setting), we assume\footnote{It is easy to see that the latter assumption can be removed without any change is our time-bounds, by finding a 2-approximation of $D$ in time $O(D)$, using the beep waves tool of~\cite{SODA-LE}.} that nodes know a polynomial upper bound on $n$ and a constant factor upper bound on diameter $D$. In the known topology setting, similar to~\cite{GPX05}, we assume that nodes know the whole graph.

\subsection{Our Results}\label{subsec:result}
Our main results are as follows:
\begin{theorem}\label{thm:singlebcast}
In radio networks with unknown topology and with collision detection, there is a randomized distributed algorithm that broadcasts a single message in $O(D + \log^6 n)$ rounds, with high probability.
\end{theorem}

\begin{theorem}\label{thm:multipleBcastKnown}
In radio networks with known topology (even without collision detection), there is a randomized distributed algorithm that broadcasts $k$ messages in $O(D + k \log n + \log^2 n)$ rounds, with high probability. 
\end{theorem}

\begin{theorem}\label{thm:multipleBcastUnknown}
In radio networks with unknown topology and with collision detection, there is a randomized distributed algorithm that broadcasts $k$ messages in $O(D + k \log n + \log^6 n)$ rounds, with high probability.
\end{theorem}

About \Cref{thm:singlebcast}, we remark that prior to this work, the best known solution for single-message broadcast was the $O(D\log{n/D}+\log^2 n)$ algorithms presented independently by Czumaj and Rytter~\cite{CR}, and Kowalski and Pelc\cite{KP}, for the model without collision detection. In that model, these bounds are optimal~\cite{ABLP, KM}. As Peleg points out in~\cite{Peleg07}, prior to this work, it was unclear whether these upper bounds can be improved in the model with collision detection. \Cref{thm:singlebcast} answers this question by showing that a better upper bound is indeed achievable. We remark that the bound of \Cref {thm:singlebcast} is within an additive poly-log of the $\Omega(D+\log^2 n)$ lower bound, that follows from the $\Omega(\log^2 n)$ lower bound of~\cite{ABLP} and the obvious lower bound of $\Omega(D)$.

About Theorems \ref{thm:multipleBcastKnown} and \ref{thm:multipleBcastUnknown}, we remark that these two results use random linear network coding (RLNC). Moreover, we note that even in the strong model of centralized algorithms with full topology knowledge, with collision detection, and with network coding, $k$-message broadcast has a lower bound of $\Omega(D+k\log n+ \log^2 n)$ rounds. This lower bound follows from the $\Omega(k \log n)$ throughput-based lower bound of~\cite{NCLB} for a $k$-message broadcast, the $\Omega(\log^2 n)$ lower bound of~\cite{ABLP} for a single message broadcast, and the trivial $\Omega(D)$ lower bound. Thus, the complexity of \Cref{thm:multipleBcastKnown} is optimal and the complexity of \Cref{thm:multipleBcastUnknown} is optimal modulo the additive $O(\log^6 n)$ term. 

When looking at the issue from a practical angle, \Cref{thm:singlebcast} and \Cref{thm:multipleBcastUnknown} have an interesting message: they show that one can replace the (expensive and not-completely-reasonable) assumption of all nodes knowing the full topology of the network, with (the considerably more reasonable and usually-available) collision detection, and still perform single or multiple broadcast tasks almost in the same time.
\smallskip

To achieve the above three results, we present three new technical elements, which each can be interesting on their own:  
\begin{enumerate}
\item[\textbf{(A)}] The first element is a distributed construction of a Gathering-Spanning-Tree (GST) with round complexity of $O(D\log^4 n)$. GSTs were first introduced by~\cite{GPX05} to obtain broadcast algorithms with an additive $O(D)$ diameter dependence in the known topology setting~\cite{GPX05, GP, FQ06}. The only known construction of GST prior to this work was the centralized algorithm of Gasieniec et al.~\cite{GPX05}, which has step-complexity of $O(n^2)$ operations and requires the full knowledge of the graph. We use our new GST construction to prove \Cref{thm:singlebcast}. For this we first decompose the graph appropriately, then we construct a GST for every part in parallel and lastly we use this setup to broadcast the (single) message efficiently. 

\item[\textbf{(B)}] The second element is a novel transmission schedule atop GST for solving multiple message broadcast problems. We contend this schedule to be the right generalization of~\cite{GPX05} for multiple messages. Such a generalization was also attempted in~\cite{FQ06} but its correctness was disproved~\cite{QinPersonalCommunication}.

\item[\textbf{(C)}] The third element is \emph{backwards analysis}, an new way to analyze the progress of messages during a multi-message radio network broadcast.
Backward analysis shows that a message spreads quickly even when other messages that are spread at the same time cause collisions. A priori it is not clear that information dissemination remains efficient in the presence of these collisions, which only arise in the mutli-message setting. Insights from the backwards analysis were crucial in the design of our multi-message transmission schedule and also enable us to apply the projection analysis of Haeupler~\cite{Haeupler11} for analyzing random linear network coding to prove \Cref{thm:multipleBcastKnown} and \Cref{thm:multipleBcastUnknown}.

\end{enumerate}

\subsection{Related Work}
Designing distributed broadcast algorithms for radio networks has received extensive attention, starting with the pioneering work of Bar-Yehuda, Goldreich and Itai (BGI)~\cite{BGI}. Here, we present a brief review of the results that directly relate to this paper.  

\paragraph{Single-Message Broadcast:} Peleg~\cite{Peleg07} provides a comprehensive survey of the known results about single-mes\-sage broadcast. BGI~\cite{BGI} present the Decay protocol which broadcasts a single message in $O(D\log n+ \log^2 n)$ rounds. The best known distributed algorithms for single-message broadcast in for the setting where the topology is unknown are the $O(D\log{\frac{n}{D}}+\log^2 n)$ algorithms presented independently by Czumaj and Rytter~\cite{CR}, and Kowalski and Pelc\cite{KP}. These algorithms can be viewed as clever optimizations of the Decay protocol~\cite{BGI}. Moreover, similar to the Decay protocol, these two algorithms are presented for the model without collision detection and are optimal in that model~\cite{ABLP, KM}. Prior to this work, no better algorithm was known for the model with collision detection. If the topology of the network is known, then the algorithm of Gasieniec, Peleg and Xin~\cite{GPX05} achieves the optimal $O(D + \log^2 n)$ time complexity. Kowlaski and Pelc~\cite{KowalskiPelc2007} gave an explicit deterministic broadcast protocol which achieves the same time complexity.

\paragraph{Multi-Message Broadcast:} The complexity of multi-mes\-sage broadcast (with bounded packet size) is less understood. In the model without collision detection, the following results are known. The earliest work on multi-message broadcast problem is by BarYehuda et al.~\cite{BII93}, which broadcasts $k$ messages in $O((n + (k + D) \log n) \log \Delta)$ rounds, where $\Delta$ is the maximum node degree. Chlebus et al.\cite{Chlebus2011} present a deterministic algorithm that broadcasts $k$ messages in $O(k\log^3 n + n\log^4 n)$ rounds. Khabbazian and Kowalski \cite{KK} and Ghaffari and Haeupler \cite{DISC-Structuring} give randomized algorithm that reduce the dependency on $k$ to $O(k \log n)$ using coding techniques. Ghaffari et al.\cite{NCLB} give an $\Omega(k \log n)$ lower bound which shows that this throughput is optimal and furthermore study whether coding is neccesary to achieve this throughput. The randomized algorithms of \cite{KK} and \cite{DISC-Structuring} broadcast $k$ messages in $O(k\log\Delta+(D+\log n)\log n\log\Delta)$ rounds and $O(k\log\Delta+(D+\log n)\log n\log\Delta)$ rounds respectively. Again, prior to this work, no better algorithm was known for the model with collision detection. 


%% file: Gathering-BFS.tex
\section{Single-Message Broadcast}\label{sec:GST}
We first recall the definition of a Gathering-Spanning-Tree (GST)~\cite{GPX05}, in \Cref{subsec:GSTdef}. Then, in \Cref{subsec:GSTAlg}, we present a distributed algorithm with time complexity $O(D \log^4 n)$ for constructing a GST, in radio networks with unknown topology (even without collision detection). In \Cref{subsec:single-message}, we then show that this algorithm can be used to broadcast a single message in $O(D+ \log^6 n)$ rounds, in radio network with unknown topology but with collision detection. 

\subsection{Gathering Spanning Trees (GST)}\label{subsec:GSTdef}
\paragraph{Ranked BFS:}
Consider a BFS tree $\mathcal{T}$ in graph $G$, rooted at source node $s$. Also, suppose that in this tree, we have assigned to each node $v$ a level number $\ell(v)$, which is equal to the distance of $v$ from $s$. We \emph{rank} the nodes of $\mathcal{T}$ using the following inductive \emph{ranking rule}: 
Each leaf of $\mathcal{T}$ gets rank $1$. Then, consider node $v$ and suppose that all children of $v$ in $\mathcal{T}$ are already ranked. Let $r$ be the maximum rank of these children. If $v$ has exactly one child with rank $r$, then node $v$ gets rank $r$. If $v$ has two or more children with rank $r$, then $v$ gets rank $r+1$. 
As shown in \cite{GPX05}, one can easily see that in each \emph{ranked BFS}, the largest rank is at most $\ceil{\log_{2} n}$.

\begin{figure*}[t]
	\centering
		\includegraphics[width=\shortOnly{0.8\textwidth}\fullOnly{0.9\textwidth}]{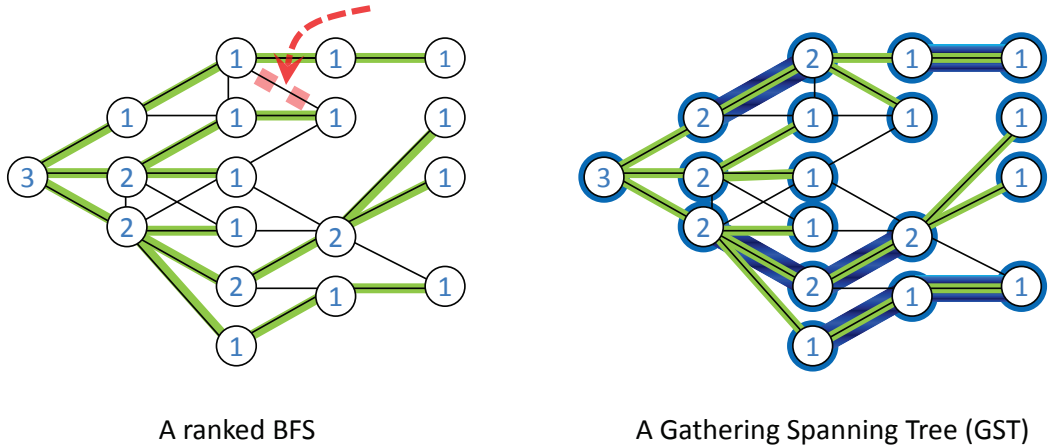}
	\caption{Gathering Spanning Tree}
	\label{fig:GST}
\end{figure*}

\paragraph{Gathering Spanning Tree (GST)\cite{GPX05}:}
A ranked BFS-tree $\mathcal{T}$ is called a GST of graph $G$ if and only if the following \emph{collision-freeness property} is satisfied: 

\bigskip

\noindent\fbox{
	\parbox{0.95\linewidth} {
In graph $G$, any node of rank $r$ on level $l$ of $\mathcal{T}$ is adjacent to \emph{at most one} node of rank $r$ at level $l-1$ of $\mathcal{T}$. In other words, if there are two nodes $u_1$ and $u_2$ with rank $r$ on level $l$ of $\mathcal{T}$, and their parents in $\mathcal{T}$ are respectively $v_1$ and $v_2\neq v_1$ (on level $l-1$ of $\mathcal{T}$), and $v_1$ and $v_2$ have rank $r$ as well, then there is no edge between $v_1$ and $u_2$ or between $v_2$ and $u_1$. 
	}
}
\medskip

\paragraph{Fast Stretches in a GST:}
In a GST $\mathcal{T}$, for each path in $\mathcal{T}$ from a node $v$ to a node $u$ that is a descendant of $v$ in $\mathcal{T}$, we call this path a \emph{fast stretch} if all the nodes on the path have the same rank. Note that a fast stretch might be just a single node.

\paragraph{Distributed GST:} In a distributed construction of a GST, each node $u$ must learn the following four items\footnote{From  (\texttt{2}) and (\texttt{4}), any node $u$ can easily infers whether it is the first node in a fast stretch and whether its parent is in that stretch as well.}: (\texttt{1}) its level $\ell(u)$,  (\texttt{2}) its own rank $r(u)$,  (\texttt{3}) the id of its parent $v$, and (\texttt{4}) the rank of its parent $r(v)$.

Figure \ref{fig:GST} presents an example of a GST. The black edges present the graph $G$ and the thicker green edges present a rank labeled BFS tree $\mathcal{T}$ of $G$. On the left side, we see a rank-labeled BFS tree, but this tree is not a GST because of the violation of the collision-freeness property indicated by the red dashed arrow. On the right side, we see another rank-labeled BFS of the same graph $G$, which is a GST. In this GST, the green edges that are coated with wide blue lines indicate the fast stretches. Each node that is not incident on any of these blue-coated edges forms a trivial fast-stretch made of just a single node.

\paragraph{Broadcast Atop GST:} In~\cite{GPX05} Gasieniec et al. presented an algorithm to broadcast a single message in $O(D+\log^2 n)$ rounds, atop a GST. A high-level explanation is as follows: with a careful timing, the message can be sent through the fast stretches without any collision. That is, we can (almost simultaneously) send the message through different stretches such that in each fast stretch, the message gets broadcast from the start of the stretch to the end of the stretch in a time asymptotically equivalent to the length of the stretch. On the other hand, since the largest rank in the tree $\mathcal{T}$ is at most $\ceil{\log_{2} n}$ and because on each path from the source to any node $v$, the ranks are non-increasing, we get that the path from the source to each node $v$ is made of at most $\ceil{\log_{2} n}$ distinct fast stretches. By using the \emph{Decay protocol}\footnote{The Decay protocol is a standard technique for coping with collisions in radio networks. We present a short recap on it in \Cref{subsub:recruit}.} \cite{BGI} on each of the (at most) $\ceil{\log_{2} n}$ connections between the fast stretches, we get a broadcast algorithm with time complexity $O(D + \log^2 n)$. We refer the reader to \cite{GPX05} for the details of this broadcast algorithm. We remark that we will use \cite{GPX05} simply as a black-box that broadcasts a single-message in time $O(D + \log^2 n)$ on top of the GSTs we construct. 

\subsection{Distributed GST Construction}\label{subsec:GSTAlg}
In this subsection, we present the following result:

\begin{theorem}\label{thm:GSTconst}
In the radio networks (even without collision detection), there exists a distributed GST construction algorithm with time complexity $O(D \log^4 n)$ rounds. 
\end{theorem}

We show a GST construction with round-complexity of $O(D \log^5 n)$ in Sections \ref{subsub:recruit} to \ref{subsub:bipartite}. We later improve this to $O(D \log^4 n)$ rounds, in \Cref{subsub:mod}. 

\subsubsection{Black-Box Tools}\label{subsub:recruit}
Before starting the construction, we first present two black-box tools which we use in our construction.

\medskip
\noindent \textbf{\emph{Decay Protocol}\cite{BGI}:} Rounds are divided into phases of $\log n$ rounds, and in the $i^{th}$ round of each phase, each node $v$ transmits with probability $2^{-i}$ (if it has a message for transmission). 
\begin{lemma}\label{lem:decay}\textbf{(Bar-Yehuda et al.\cite{BGI})}
For each node $v$, if at least one neighbor of $v$ has a message for transmission, then in each phase of the Decay protocol, node $v$ receives at least one message with probability at least $\frac{1}{8}$. Moreover, in $\Theta(\log n)$ such phases, $v$ receives at least one message, with high probability.
\end{lemma}

\bigskip
\noindent \textbf{\emph{Recruiting Protocol}:} This tool can be abstracted by the guarantees that it provides, which we present in \Cref{lem:recruit}. 

\begin{lemma}\label{lem:recruit} Consider a bipartite graph $\mathcal{H}$ where nodes on one side are called \emph{red} and nodes on the other side are called \emph{blue}. The recruiting protocol achieves the following three properties, w.h.p., in $\Theta(\log^3 n)$ rounds:
\begin{inparaenum}[\bfseries{} (a)] 
\item for each blue node $u$, we assign an adjacent red node of $v$ to $u$. In this case, we say $u$ is \emph{recruited} by $v$ (then called parent of $u$), 
\item each red node $v$ knows whether it recruited zero, one, or at least two blue nodes, 
\item each recruited blue node $u$ knows whether its parent $v$ recruited zero, one, or at least two blue nodes.
\end{inparaenum} 
\end{lemma}


\medskip\smallskip

\noindent\fbox{
\vspace{0.1cm}

	\parbox{0.95\linewidth} {
	\smallskip
		\textbf{{Recruiting Protocol:}} The protocol consists of $\Theta(\log^2 n)$ recruiting iterations, each having $2+\Theta(\log n)$ rounds as follows: 
		
		\medskip
		\begin{adjustwidth}{0.0em}{0em}
		\begin{itemize}
		\item In the first round of the $j^{th}$ recruiting iteration, each red node transmits its id with probability $2^{-\lceil\frac{j}{\Theta(\log n)}\rceil}$.
		
		\item Then, we run a phase of the Decay protocol, consisting of $\Theta(\log n)$ rounds, from the side of blue node. In this phase, each not-recruited blue node $u$ that received a message of a red node $v$ tries to transmit $u.id$ and $v.id$ (together in one packet). 
		
		\item After that, the red nodes repeat the exact transmissions of the first round of this iteration, with new contents as follows: (1) if in the previous Decay phase, a red node $v$ received its own id from exactly one blue node $u$, then $v$ broadcasts $v.id$, (2) if the red node $v$ received its own id from two or more blue nodes, then $v$ broadcasts a special message $\Sigma$. (3) Otherwise, $v$ transmits an empty message. 

	\item Next, if a blue node $u$ received its own id or the special message $\Sigma$ in the last round, then we say $u$ is recruited by red node $v$, where $v$ is the red node such that $u$ received $v.id$ in the first round. Note that each red node $v$ knows whether it recruited zero, one or at least two blue nodes. 
	\end{itemize}
	\end{adjustwidth}
		\vspace{0.1cm}
	}
}

\bigskip

\fullOnly{
\begin{proof}[Proof of \Cref{lem:recruit}] We show that each blue node is recruited with high probability. The other parts follow easily from the description of the algorithm 

Consider an arbitrary blue node $u$. It is easy to see that there are $\Theta(\log n)$ iterations such that in the first round of each of these iterations, $u$ receives the message of a red node. This is because, for each $j^{th}$ iteration where $j \in [\frac{d(u)\,\cdot\,\log n}{2},\; 2d(u)\,\cdot\,\log n]$, where $d(u)$ is the degree of $u$ in $\mathcal{H}$, $u$ receives a message in the first round of iteration $j$ with constant probability. A Chernoff bound then shows that in $\Theta(\log n)$ of these iterations, in the first round, $u$ receives the message of a red node.

Consider one such recruiting iteration, and suppose that in the related first round, $u$ receives the message of red node $v$. In the $\Theta(\log n)$ rounds of the Decay phase of that iteration, from the properties of the Decay protocol, we get that with constant probability, the red node $v$ either receives the message of $u$ or it receives at least two messages from blue nodes. Moreover, if $v$ receives a message from a blue node $w$, then $w$ had received the message of node $v$ in the first round of this iteration. This is because, since $v$ transmitted in that round, $w$ could not have received from any other red node $v'$ and since $w$ is transmitting in the decay, we know that it has received the message of one red node. Thus, we conclude that with constant probability, the red node $v$ receives either the message of $u$ or at least two messages from blue nodes. In either case, $u$ gets recruited. Note that $u$ received the message of $v$ in the last round of the iteration simply because this round is an exact repetition of the transmission of the first round of this iteration, where $u$ received a message from $u$.

Now in $\Theta(\log^2 n)$ recruiting iterations, there are $\Theta(\log n)$ iterations where in their first round, $u$ receives the message of a red node. Since in each such iteration $u$ is recruited with a constant probability, we get that after the full run of the Recruiting protocol, $u$ is recruited with high probability.
\end{proof}
}

\subsubsection{GST Construction Outline} 
We first construct a BFS-tree of $G$ and assign to each node $v$ a level $\ell(v)$ that is equal to the distance of $v$ from the source. This can be done in $O(D\log^2 n)$ rounds, as follows: Rounds are divided into $D$ epochs each consisting of $\Theta(\log n)$ phases of the Decay protocol (thus, each epoch has $\Theta(\log^2 n)$ rounds). In each epoch, a node $v$ participates in the decays if and only if it is the source or it has received a message by the end of the last epoch. During these rounds, each node relays the first message it received. The epoch in which a node $v$ receives a message for the first time determines the BFS level $\ell(v)$ of node $v$.

\smallskip
Now that we have a BFS-tree, we build the GST on top of this BFS layering, level by level, and from the largest level towards the source. For this, the problem boils down to the following scenario: Consider level $l$ of layering and assume that the GST is already built for levels $j \geq l$. Consider the bipartite graph $H$ induced on the nodes of level $l-1$ and level $l$, ignoring the (possible) edges inside each level. The core of the problem is to design an algorithm to construct the part of GST between levels $l-1$ and $l$, i.e., the part that is $H$. 

\smallskip
Let us call the nodes on level $l-1$ \emph{red nodes}, and the nodes on level $l$ \emph{blue nodes}. To construct the part of GST that is in $H$, we assign a red \emph{parent} $v$ to each blue node $u$, from amongst the red neighbors of $u$ in $H$. In this case, $v$ is known as $u$'s \emph{parent} and $u$ is a \emph{child} of $v$. This assignment, along with the rankings of blue nodes, leads to a ranking for the red nodes. More precisely, let $v$ be a red node and let $i$ be the maximum rank of blue node children of $v$ in the assignment. Node $v$ gets rank $i$ if it has only one child with rank $i$, and $v$ gets rank $i+1$ if it has more than one child with rank $i$. 

\smallskip
To have a GST, these assignments should be \emph{collision-free}. That is, if there exist blue nodes $u_1$ and $u_2$ and their respective parents $v_1$ and $v_2$, all four with rank $i$, then $H$ must have no edge between $v_1$ and $u_2$, or between $v_2$ and $u_1$.\fullOnly{ Mathematically, if we let $\mathcal{M}$ be the set of edges between blue nodes $u$ of rank $i$ and their respective red parents $v$ with rank $i$, then $\mathcal{M}$ should be an \emph{induced matching} of graph $H$.} We refer to the problem of finding such an assignment as the \emph{Bipartite Assignment Problem}. 

More precisely, in the \emph{Bipartite Assignment Problem}, we should achieve the following 6 properties: (\texttt{1}) For each blue node $u$, we should assign a red neighbor $v$ as its parent, (\texttt{2}) we should rank the red nodes as follows: for each red node $v$, suppose $i$ is the maximum rank of the children of $v$. Then, $v$ should get rank $i$ if $v$ has exactly one blue child of rank $i$, and $v$ should receive rank of $i+1$ if $v$ has two or more blue children of rank $i$, (\texttt{3}) the assignment should be \emph{collision-free}, (\texttt{4}) each red node must know its rank and (\texttt{5}) each blue node $u$ should know the id of its parent and (\texttt{6}) each blue node $u$ should know the rank of its parent.

\smallskip
The \emph{Bipartite Assignment Problem} is the core of the GST construction and once we have a solution for it, repeating the solution level by level from the largest level to source constructs a GST. In the next subsection, we explain how to solve this problem in $O(\log^5 n)$ rounds.

\subsubsection{The Bipartite Assignment Algortihm}\label{subsub:bipartite} 
Consider bipartite graph $H$ as explained. We solve the bipartite assignment problem (defined above) in $H$ in a rank by rank basis, starting with the largest possible rank $\ceil{\log n}$ (of blue nodes), and going down in ranks until reaching rank $1$. We spend $\Theta(\log^4 n)$ rounds on each rank. Let us consider the case of a bipartite assignment for blue nodes of rank $i$ in graph $H$, assuming that ranks greater than $i$ are already solved. 

We first identify the red neighbors of the blue nodes with rank $i$. This is done by using $\Theta(\log n)$ phases of the Decay protocol where blue nodes of rank $i$ transmit. This identifies the desired red nodes as every such red node receives at least one message with high probability and no other red node receives any message. From now on, throughout the procedure for rank $i$, only these red nodes are active. Now the algorithm is divided into $\Theta(\log n)$ epochs. Each epoch consists of three stages as follows:

\medskip

\begin{adjustwidth}{0.0em}{0em}
\begin{itemize}
\item[\textbf{Stage I}:] Call a blue node $u$ of rank $i$ a \emph{loner} if $u$ has exactly one active red neighbor. We first detect the loner blue nodes. For this, in one round, each active red node transmits a message. Only loner blue nodes receive a message and each other blue node receives a collision. We then use $\Theta(\log n)$ phases of the Decay protocol, where each blue loner tries transmitting. This with high probability informs all red nodes that are connected to at least one loner blue node. We call these red nodes \emph{loner-parents}. 

\medskip
\item[\textbf{Stage II}:] This stage is divided into three parts, and each red node is active in only one of the parts. Loner-parents, which we identified in the stage I, are active only in part 1. Each other active red node randomly and uniformly decides to be either \emph{brisk} or \emph{lazy}, which respectively mean it is active in part 2 or in part 3. These parts are as follows:  

\medskip
\begin{adjustwidth}{1.0em}{0em}
\begin{enumerate}
\item[\textbf{Part 1.}] Loner-parents use a \emph{recruiting protocol}. During this recruiting protocol, each blue neighbor of each red loner-parent get recruited with high probability. These assignments are \emph{permanent}. All the blue nodes that are recruited become inactive for the rest of the assignment problem.

\item[\textbf{Part 2.}] Brisk red nodes run a Recruiting protocol. Then, each blue node that is not the only recruited child of its parent considers its parent as its \emph{permanent} GST parent and becomes inactive permanently (for the GST construction). The other recruited blue nodes become inactive only for the remainder of this epoch, but these assignments are \emph{temporary} and the related nodes restart in the next epoch, ignoring their temporary assignments.
\item[\textbf{Part 3.}] We repeat the procedure of part 2, but this time with lazy red nodes and with the active blue nodes that did not get recruited in parts 1 or 2.
\end{enumerate}
\end{adjustwidth}

\medskip
\item[\textbf{Stage III}:] Let us say that a red node is \emph{marked} if it was a loner-parent or if it recruited zero or strictly more than one blue nodes in parts 2 or 3. Each marked red node becomes inactive after this epoch. Thus, the only red nodes that remain active after this epoch are those that do not have any loner neighbor and recruited exactly one child in part 2 or 3 of the stage II. Each marked red node knows whether it recruited zero, one, or at least two children (in stage II). We use this knowledge to rank these marked red nodes giving them rank of $i$ if they recruited exactly one blue child and rank of $i+1$ if they recruited more than one blue child. Blue children of marked red nodes also know that their parents of marked and they can also compute the rank of their parents (refer to property (c) of \Cref{lem:recruit}).

\ \ \ Before inactivating the marked red nodes, we do one simple thing: marked red nodes run $\Theta(\log n)$ phases of the Decay protocol sending their id and rank. Each blue node of any rank strictly lower than $i$ that receives a red node id considers the first red node that it heard from as its permanent GST parent, records the id and rank of that red parent, and then, becomes inactive for the rest of the assignment problem.

\end{itemize}
\end{adjustwidth}
\smallskip

\noindent After running the bipartite assignment algorithm for all the ranks, if a red node $v$ has no child, then $v$ is a leaf and in the GST, $v$ gets rank $1$.
\medskip
\bigskip

Figure \ref{fig:Assignment} shows an example of assignments during an epoch (the first epoch). The green arrows in the leftmost part indicate the loner blue nodes at the start of the epoch. The loner parent red nodes are indicated by a number $1$ next to them, meaning they are active in part 1. Brisk and lazy red nodes are respectively indicated by numbers $2$ and  $3$, next to them. The smaller nodes present the (temporarily or permanently) deactivated nodes. The green dashed lines show the permanent assignments and the (thicker) orange dashed lines show the temporary assignments. After the end of epoch, nodes with temporary assignment are re-activated. The graph remaining after the first epoch is presented on the right side of the Figure \ref{fig:Assignment}, by solid blue lines.

\begin{figure*}[t]
	\centering
		\includegraphics[width=\shortOnly{0.98\textwidth}\fullOnly{\textwidth}]{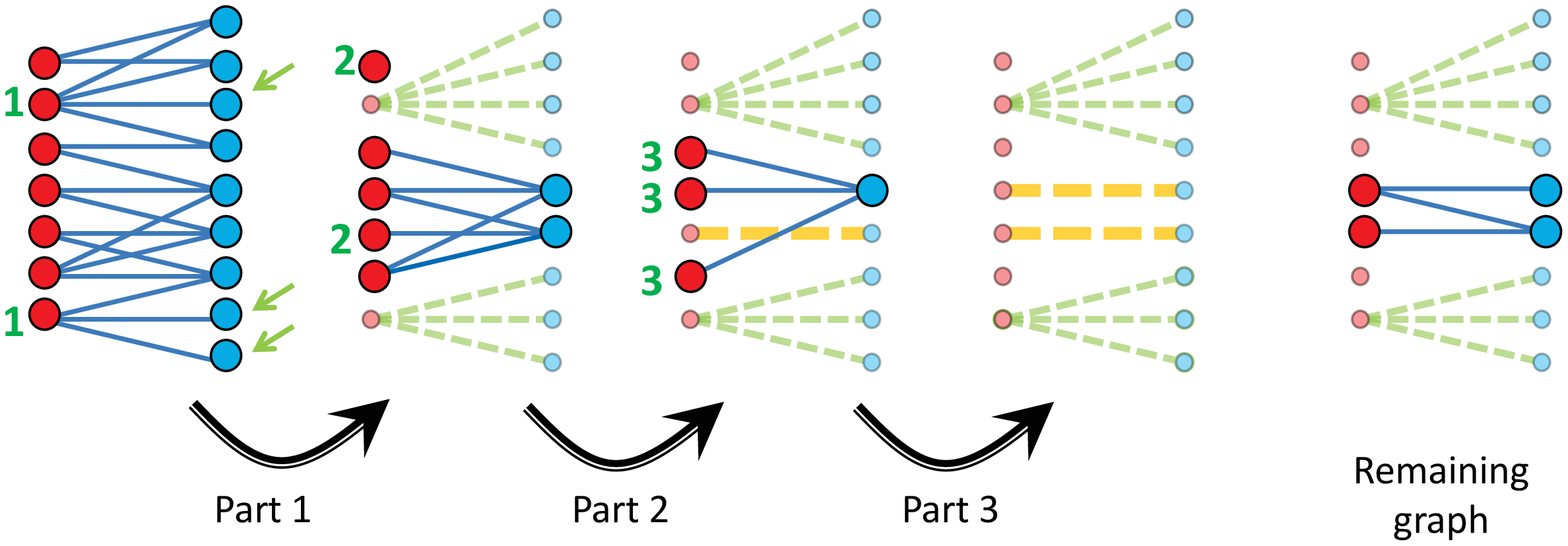}
	\caption{Parts $1$, $2$, and $3$ of the stage II of the first epoch of the assignment algorithm, and the graph remaining after the first epoch}
	\label{fig:Assignment}
\end{figure*}


\medskip
\paragraph{Analysis:} In \Cref{{lem:expShrinkage}}, we prove that in each of the $\Theta(\log n)$ epochs except the first one, we reduce the size of the assignment problem for rank $i$ by at least a constant factor, with at least a positive constant probability. Here, by size of the assignment problem, we mean the number of the active red nodes with a blue neighbor of rank $i$. A standard Chernoff bound then shows that in $\Theta(\log n)$ epochs, each blue node of rank $i$ has a parent. It is clear that the parents are ranked according to the ranking rules of GST and nodes know their own rank, the id of their parents, and the rank of their parents. We show in \Cref{lem:GSTcollisiotn-free} that with high probability, the assignment is collision-free. 
\begin{lemma}\label{lem:expShrinkage} In each epoch $j' \leq 2$, with a probability at least $1/7$, the number of remaining active red nodes for the next epoch goes down with a factor at least $8/7$.
\end{lemma}

\begin{proof}Consider epoch $j' \geq 2$ and let $\eta$ be the number of active red nodes at the start of this epoch. We show that the expected number of red nodes that remain active at the end of this epoch is at most $\frac{3\eta}{4}$. This is enough for the proof because with this, and by Markov's inequality, we get that with probability at least $1/7$, the number of active remaining red nodes at the end of this epoch is at most $\frac{7\eta}{8}$. 

Each red node remains active after epoch $j'$ only if it gets a temporary assignment, i.e., if it is not a loner-parent and it recruits exactly one child during parts 2 and 3 of Stage II. Thus, the expected number of red nodes that remain active is at most equal to the expected of number of brisk red nodes (those that act in part 2) plus the number of blue nodes that are active in part 3. The expected number of brisk red nodes is at most $\frac{\eta}{2}$. To complete the proof, we show that the expected number of blue nodes that remain active for part 3 (after the assignments of part 2) is at most $\frac{\eta}{4}$. 


After each epoch, the only red nodes that remain active are those that have a temporary assignment, i.e., those that each have recruited exactly one child and that child is not a loner. Moreover, the only active remaining blue nodes are those blue nodes temporarily matched to the remaining red nodes. Thus, after each epoch, the number of remaining active red nodes and the number of remaining active blue nodes are equal. From this, we can conclude that since $j'\geq 2$, at the start of epoch $j'$, the number of active blue nodes is at most $\eta$. 

Using \Cref{lem:recruit}, we infer that in part 1 of stage II, each blue neighbor of a loner-parent is w.h.p. recruited by a red loner-parent. Thus, in particular, each loner is recruited with high probability. Hence, at the start of part 2 of stage II, each remaining active blue node has at least $2$ red node neighbors. Since each non-loner-parent red node is active in part 2 of stage II with probability $1/2$, and because in part 2 of stage II each active blue node that has an active red node neighbor gets recruited with high probability (by \Cref{lem:recruit}), each blue node remains active after part 2 of stage II with probability at most $1/4$.  We know that because of the previous paragraph, the number of active remaining blue nodes at the start of part 2 of stage II is at most $\eta$. Hence, the expected number of blue nodes remaining active after part 2 is at most $\frac{\eta}{4}$. This completes the proof of the lemma.
\end{proof}

\begin{lemma}\label{lem:GSTcollisiotn-free} With high probability, the bipartite assignment algorithm creates a collision-free assignment.
\end{lemma}

\begin{proof} We show that if there exist blue nodes $u_1$ and $u_2$ ($u_1 \neq u_2$) and their respective red parents $v_1$ and $v_2$ ($v_1 \neq v_2$), all four with rank $i$, then with high probability, $H$ must not have any edge between $u_2$ and $v_1$, or between $u_1$ and $v_2$. For the sake of contradiction, and without loss of generality, suppose that there is an edge between $u_2$ and $v_1$. \fullOnly{\Cref{fig:Collision} shows the configuration of these four nodes.
\begin{figure}[h]
	\centering
		\includegraphics[width=0.35\textwidth]{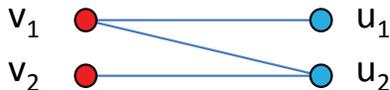}
	\caption{Collision-freeness proof}
	\label{fig:Collision}
\end{figure}
}
Since $v_2$ and $u_2$ have rank $i$, blue node $u_2$ must have been a loner when $v_2$ recruited it. Thus, $v_2$ recruited $u_2$ after $v_1$ became inactive. Hence, in the epoch that $v_1$ recruited $u_1$, $u_2$ was active. Therefore, using \Cref{lem:recruit} we get that in the part 1 of the epoch in which $v_1$ recruited $u_1$, $u_2$ must have been w.h.p. recruited by either $v_1$ or some other loner-parent. Since $v_2\neq v_1$ recruited $u_2$, we get that $v_2$ must have been that other loner parent. This means that at that time, $v_2$ had a loner child ($\neq u_2$) and thus, $v_2$ has recruited more than one child of rank $i$. This means that $v_2$ must have had rank $i+1$ which contradicts with the assumption that $v_2$ has rank $i$. 
\end{proof}

\subsubsection{Pipelining the GST Construction}\label{subsub:mod} Note that in the algorithm described in \Cref{subsub:bipartite} where we are working on the assignment problem between levels $l-1$ and $l$, once we are done with the assignment problem of ranks $i$ and $i-1$, nodes of level $l-1$ that receive rank $i$ are already determined, i.e., no other node in level $l-1$ will receive rank $i$. Thus, we can solve the two problems of rank $i-2$ assignment between levels $l-1$ and $l$ and rank $i$ assignments between levels $l-2$ and $l-1$, essentially simultaneously, by interleaving them in even and odd rounds. Using the same idea, it is easy to see that one can pipe-line the assignment problems of different ranks between different levels. Then, the assignment problem between levels $l-1$ and $l$ starts after $\Theta((D-l) \log^4 n )$ rounds. Thus, the assignment problem of largest possible rank between levels $0$ and $1$ starts after $\Theta(D \log^4 n)$ rounds. The largest rank is at most $\ceil{\log n}$. Since each rank takes $\Theta(\log^4 n)$ rounds, the whole GST construction problem finishes after $\Theta(D \log^4{n})$ rounds.


\subsection{Unknown Topology Single-Message Broadcast in $O(D + \log^6 n)$ Rounds}\label{subsec:single-message}

\begin{theoremR}{thm:singlebcast}
In radio networks with unknown topology and with collision detection, there is a randomized distributed algorithm that broadcasts a single message in $O(D + \log^6 n)$ rounds, with high probability.
\end{theoremR}
\begin{proof}
We first use a wave of collisions to get a BFS layering in time $D$. That is, the source transmits in all rounds $[1, D]$, and each node $v$ transmits in all rounds $[r, D]$ where $r$ is such that $v$ receives a message or a collision in round $r-1$. For each node $v$, the round $r-1$ in which $v$ receives the first message or collision determines distance of $v$ from the source. 

\smallskip
Having this BFS layering, we decompose the graph into $O(\log^4 n)$ rings, each consisting of $D' = {D}/{\log^4 n}$ consecutive layers of the BFS layering. 

\smallskip
Then, we compute a gathering spanning tree for each of the rings in $O(D' \log^4 n) = O(D)$ rounds. Note that computation of a GST for each ring only depends on $D'$ which is the number of BFS layers that the ring contains, and that given the BFS-layering, the computation of the GSTs of all rings is performed in parallel. 

\smallskip
Having these GSTs, broadcasting the message inside each ring takes $O(D'+\log^2 n)$ rounds, using \cite{GPX05}. Finally, we use $O(\log^2 n)$ rounds of the Decay protocol \cite{BGI} to propagate the message from the outer boundary of one ring to the inner boundary of the next ring. Since there are $O(\log^4 n)$ rings, the whole broadcast takes $\big(O(D'+\log^2 n) + O(\log^2 n)\big) \cdot O(\log^4 n) = O(D+ \log^6 n)$ rounds.
\end{proof}


%% file: MMS.tex
\section{Multi-Message Broadcast}\label{sec:MultiMessage}


In this section, we show the following two results:

\begin{theoremR}{thm:multipleBcastKnown}
In radio network with known topology (even without collision detection), there is a randomized distributed algorithm that broadcasts $k$ messages in $O(D + k \log n + \log^2 n)$ rounds, with high probability. 
\end{theoremR}

\begin{theoremR}{thm:multipleBcastUnknown}
In radio networks with unknown topology and with collision detection, there is a randomized distributed algorithm that broadcasts $k$ messages in $O(D + k \log n + \log^6 n)$ rounds, with high probability.
\end{theoremR}

In Subsections \ref{subsec:challenges} to \ref{subsec:analysisMMBA}, we present and analyze the algorithm that achieves \Cref{thm:multipleBcastKnown}. We remark that the $O(D + k \log n + \log^2 n)$ round-complexity of \Cref{thm:multipleBcastKnown} is optimal, given the $\Omega(k \log n)$ lower bound of \cite{NCLB} for $k$-message broadcast, the $\Omega(\log^2 n)$ lower bound of \cite{ABLP} for single message broadcast, and the trivial $\Omega(D)$ lower bound. 

\smallskip
Furthermore, it is easy to combine the known topology algorithm of \Cref{{thm:multipleBcastKnown}} with the ideas of the proof of \Cref{thm:singlebcast} (i.e., breaking the graph into rings of radius $\ceil{\frac{D}{\log^4 n}}$) and the standard technique of grouping messages and pipe-lining the groups, to prove \Cref{thm:multipleBcastUnknown}. We present the details of this part in \Cref{subsec:MMB-Unknown}.


\subsection{Challenges in Broadcasting Multiple Messages}\label{subsec:challenges}

Given the known transmission schedules for broadcasting a single message in optimal $O(D+\log^2 n)$ time on top of a GST, it is intriguing to try to use the same transmission schedule to solve the multi-message broadcast problem. However, since we cannot disjoin the spreading process of different messages, this approach faces two challenges:

Firstly, when a node $v$ has already learned multiple messages and is triggered by the schedule to transmit, $v$ needs to decide which message to forward. Choosing one message over the others can slow down the progress of those other messages. Fortunately, random linear network coding (RLNC) \cite{HKMKE} provides a general technique for making such decisions: Instead of deciding on one specific message whenever $v$ is triggered to send, it transmits a random linear combination of all packets it has received. It has been shown that this is the universal optimal strategy, that is, this succeeds with high probability as soon as it was possible (in hindsight) to send $k$ messages to each of the receivers~\cite{OptNC}. There are furthermore indications that network coding might be necessary for obtaining an asymptotically optimal throughput performance~\cite{NCLB}. Our multi-message broadcast utilizes RLNC and uses recent advances in analyzing RLNC performance~\cite{Haeupler11} for the proofs. Even though RLNC and its analysis need to be carefully tailored to the radio broadcast setting here, this already gives us a good plan to remedy the first issue.

The second issue is subtle but turns out to be more problematic: When proving progress of messages, all known single-message schedules and their analyses (e.g., those of~\cite{GPX05}) rely crucially on the fact that the nodes that do not have the (single) message remain silent and cause no collisions. In a multi-message setting it becomes a necessity that we make progress for a message while allowing other nodes that do not have this message to transmit (in order to make progress on other messages). 

Trying to understand and resolve this problem prompted us to define the property of a transmission schedule being \emph{multi-message viable (MMV)}: 
\begin{definition}
We say that a transmission schedule broadcasts one message in a multi-message viable (MMV) way in $T$ rounds with probability $1-\delta$ if the following holds: Suppose that we use this transmission schedule but nodes that do not have the message but are scheduled to transmit send ``noise''. Then, the message is broadcast to all nodes in $T$ rounds with probability $1-\delta$. 
\end{definition}
Intuitively, this notion captures the viewpoint where we focus on one message and the transmissions of the other messages are regarded as noise, possibly harming the progress of the message in consideration. We later see that this notion is enough to prove that a schedule works well with RLNC. 

Unfortunately proving that a schedule is MMV is not straightforward and it is a priori not clear whether the already existing schedules are MMV. The easiest example to see this is the well-known Decay protocol of~\cite{BGI}: in the classical implementation of the Decay protocol, if a node is scheduled to transmit but it does not have the message, then this node remain silent. The Decay protocol broadcasts a single message in $O(D\log n+\log^2 n)$ rounds, with high probability~\cite{BGI}. This follows almost directly from a simple progress lemma which shows that in $O(\log n)$ rounds of the protocol, a node receives the message with constant probability if at least one of its neighbors already has the message. However, if the nodes that do not have the message are allowed to send noise when the schedule prompts them to transmit, then this key progress lemma of \cite{BGI} does not hold anymore. Surprisingly, even though the progress lemma breaks, it is still true that one message is spread quickly in this case (when nodes that do not have the message are noising), meaning that the Decay protocol broadcasts in time $O(D\log n+\log^2 n)$ rounds, w.h.p., in an MMV way: 

Before formally proving this fact, first let us recall the details of the transmission schedule of the Decay protocol: 

\medskip
\noindent\fbox{
	\parbox{0.95\linewidth} {
\textbf{Transmission Schedule of the Decay protocol in an MMV Framework}: For each round $r$, for each node $v$ at distance $l_v$ from source, if $r\equiv l_v+1 \mod 3$, then $v$ is prompted to transmit with probability $2^{-((r -l_v-1)/3 \mod {\ceil{\log n}})}$. If $v$ is prompted but does not have the message, it sends ``noise''.
	}
}

\smallskip

\begin{lemma} \label{lem:DecayMMV} The Decay protocol broadcasts one message in an MMV way in $O(D\log n+\log^2 n)$ rounds, w.h.p.
\end{lemma}

To prove this lemma, we need to go away from the analysis approach in~\cite{BGI} which chooses a shortest path from source $s$ to node $v$ and shows that the broadcast message makes fast progresses along this path when moving forwards in time. Instead we use what we call \emph{backwards analysis}: In a nutshell, we move backwards in time and find a sequence of collision-free transmissions from $s$ to $v$, where hops of this sequence are unraveled backwards (from $v$ to $s$). Meanwhile unraveling this sequence, each of these transmission can be the broadcast message or just ``noise'', depending on whether the sender has received the broadcast message or not. Once we reach $s$, it means the transmissions in the sequence indeed where the broadcast message. 

\begin{proof}[Proof of \Cref{lem:DecayMMV}]
Fix an arbitrary node $v$. Let $T=\lambda(D\log n+\log^2 n)$ for a large enough constant $\lambda$. For each integer $t$, we say node \emph{``$u$ is transmission-connected to $v$ by backwards time $t$"} if there is a timely sequence of transmissions $u=w_1, w_2, \dots w_\ell=v$ where for each $i\in [1, \ell-1]$, $w_i$ transmits in a round $r_i \in[T-t, T]$, we have $r_{i} < r_{i+1}$, and in round $r_i$ where $w_i$ transmits, $w_{i+1}$ receives a message from $w_i$. We emphasize that these transmission do not consider where the transmitted message is just ``noise'' or it is the actual message of the broadcast problem. If node $w_i$ has received the message of broadcast by the end of round $r_i-1$, then the transmission of $w_i$ in round $t_i$ is the actual message of the broadcast; otherwise, it is noise. Let $S_t(v)$, or simply $S_t$, be the set of all nodes that are transmission-connected to $v$ by backwards time $t$. For each backwards time $t$, define potential $\Phi(t) = \min_{u\in S_t} dist_{G}(s, u)$. We claim that ``\emph{for each two backwards times $t, t'>t$ such that $t'-t = 3\ceil{\log n}$, if $\Phi(t)\geq 1$, with probability at least $1/(2e)$, we have $\Phi(t') \leq \Phi(t)-1$}". A Chernoff bound then shows that with high probability $\Phi(T)=0$ meaning $s \in S_{T}$. This shows that, with high probability, there exists a sequence of collision-free transmissions (and message receptions) which starts in source $s$ and ends in node $v$ by time $T$, proving that, with high probability, $v$ receives the message of $s$ by time $T$. 

To prove the claim, consider two times $t, t'>t$ such that $t'-t = 3\ceil{\log n}$ and $\Phi(t)\geq 1$. Let $u^*$ be a node $u$ in $S_t$ that minimizes $dist_{G}(s, u)$. We show that in round interval $[T-t', T-t]$, with probability at least $1/(8)$, $u^*$ receives at least one message (be it noise or the actual broadcast message) from a neighbor $u'$ such that $dist(s,u') = dist(s,u^*)-1$. Let $k$ be the number of neighbors $u'$ of $u^*$ such that $dist(s,u') = dist(s,u^*)-1$. Consider the round $r^* \in [T-t', T-t]$ such that $(r^*-dist(s,u^*))/3 \equiv \ceil{k} \mod \ceil{\log n} $. In that rounds, only the only neighbors of $u^*$ that can transmit are those neighbors $u'$ that have $dist(s,u') = dist(s,u^*)-1$. The probability that $u^*$ receives a message from one of them is $\frac{k}{2^{-\ceil{k}}} (1-\frac{1}{2^{-\ceil{k}}})^{k-1} \geq \frac{1}{8}$. This proves the claim. 

A union bound over all nodes $v$ shows that with high probability, all nodes receives the message by round  $O(D\log n+\log^2 n)$.
\end{proof}

Unfortunately, in contrast to the transmission schedule of the Decay protocol, the GST based schedule of~\cite{GPX05} appears to be not MMV. In \Cref{sec:schedule}, we present a new transmission schedule for GSTs and again use our backwards analysis to show that this schedule is MMV. Lastly, we show that if one combines RLNC with this new schedule, then the MMV property almost directly translates into having a high broadcast throughput, leading to the optimal broadcast time of $O(D + k\log n + \log^2 n)$ rounds for $k$ messages. 

 \subsection{A Multi-Message Transmission Schedule Atop GST}\label{sec:schedule}
In this section, we present our transmission schedule for GSTs and show that it is MMV. Later we use this schedule along with random linear network coding to achieve our optimal multi-message algorithm.

\subsubsection{The Schedule}

Suppose we have a GST $T$ for graph $G$. For each node $u$, let $l_u$ be the distance of $u$ from source $s$ in graph $G$ (that is, the BFS level of $u$). Also, let $r_u$ be the rank of $u$ in GST $T$. We first construct a virtual directed graph $G'$, from graph $G$, as follows: we add a directed edge from every node $u$ with rank $r$ that is the first node of a fast stretch to every descendant of $u$ in $T$ that has rank $r$ (thus, to all nodes in that fast stretch). We call this a \emph{fast} edge. We use the notation $d_u$ to denote the length of the shortest (directed) path from $s$ to $u$ in $G'$, and we call this \emph{virtual-distance}. Given graph $G$, GST $T$, and the respective virtual graph $G'$ (and the related virtual-distances), our schedule is defined as follows:

\bigskip

\noindent\fbox{\parbox{0.95\linewidth}{
\textbf{Multi-Message Viable GST Schedule}: In round $t$, each node $u$ at BFS-level $l$ of $G$ with rank $r$ in GST $T$ and virtual-distance $d$ in the virtual graph $G'$ does as follows: 
(a) if $t \equiv 2(l + 3r) \pmod{6\ceil{\log_2 n}}$, then $u$ transmits; (b) if $t \equiv 1 + 2d \pmod{6})$, then $u$ transmits with probability $2^{-((t-1-2d)/6 \mod \ceil{\log_2 n})}$; otherwise, $u$ listens.
}}

\medskip

Note that the case (a) only happens in even rounds and case (b) happens only in odd rounds. As in~\cite{GPX05}, we call the transmissions triggered by case (a) \emph{fast transmissions} and the transmissions triggered by case (b) \emph{slow transmissions}. 

We remark that this schedule uses fast transmissions exactly as in~\cite{GPX05,FQ06} to pipeline the messages along the fast stretches of GST. We see in \Cref{lem:fastcollisionfree} that these fast transmissions are collision-free. The crucial difference with the schedule in \cite{GPX05,FQ06} lies in defining the slow transmissions with respect to the virtual-distance in graph $G'$ (instead of levels in $G$). This change results in slow transmissions not trying to push messages away from the source, but instead trying to push messages towards entry points of fast stretches (even if this leads to the message going back towards the source). While this modification seems minor, it is crucial for allowing the \emph{backwards analysis} technique to show that the new schedule is efficient and MMV. 

\subsubsection{The Analysis}

The rest of this section is dedicated to prove that the newly defined schedule is MMV:

\begin{lemma} \label{lem:GST-MMV} The MMV-GST schedule of \Cref{sec:schedule} broadcasts one message in an MMV way, in $O(D + \log n \cdot(\log n+ \log \frac{1}{\delta}))$ rounds, with probability $1-\delta$.
\end{lemma}

Before diving directly into the proof of \Cref{lem:GST-MMV} we show a few helpful invariants.

\begin{lemma} \label{lem:virtualDistance} In virtual graph $G'$, for each node $u$, we have $d_u \leq 2\ceil{\log_2 n}$.
\end{lemma}
\fullOnly{
\begin{proof}Consider the path from $u$ to $s$ in $T$. On this path, the rank never decreases, thus increases at most $\ceil{\log_2 n}$ times. Furthermore, every stretch on which the rank stays the same corresponds to a directed link in $G'$. Using this, we get the path of length at most $2\ceil{\log_2 n}$ from $s$ to $u$ in $G'$. 
\end{proof}
}

\begin{lemma}\label{lem:fastcollisionfree}
There are no collisions between any two fast transmissions.
\end{lemma}
\fullOnly{
\begin{proof}
Since fast and slow transmission happen during even and odd rounds, respectively, it is clear that collisions can only happen between two slow or two fast transmissions. To see that two fast transmission do not collide, we note that in round $t$, only nodes with a level $l \equiv t/2 \pmod{3}$ have transmissions. This is because a fast transmission in round $t$ happens only if $t \equiv 2l + 6r \equiv 2l \pmod{6}$. Since nodes whose levels differ by at least $3$ can not share a neighbor, we get that collisions can only caused by transmissions of nodes within the same level. Furthermore, two nodes within the same level are only performing a fast transmission if their ranks $r$ and $r'$ are equivalent modulo $\ceil{\log_2 n}$. By definition of GST, this implies that their ranks are equal and the \emph{collision-freeness} property of GST then guarantees that two such nodes do not share a neighbor in the next level. This shows that there are indeed no collisions between any two fast transmissions.
\end{proof}
}

\begin{proposition}\label{lem:progfast}
If node $u$ with level $l$ is the beginning of a fast stretch in GST $\mathcal{T}$ and $u$ sends a message at time $t$ in a fast transmission round, then any node $v$ with level $l'>l$ on the same fast stretch receives this message by time $t' = t + 2(l'-l)$.
\end{proposition}

\begin{lemma}\label{lem:progslow}
For any node $u$ with virtual-distance $d_u$, if there is at least one node $v$ connected to $u$ in $G$ with virtual-distance $d_v=d_u-1$, then during each interval of $6\ceil{\log_2 n}$ rounds, with probability at least $\frac{1}{8}$, node $u$ receives a message from one node with virtual-distance $d_u-1$.
\end{lemma}

\begin{proof}
Let $x$ be the number of neighbors of $u$ with virtual-distance $d-1$. Note that within any span of $6 \ceil{\log_2 n}$ rounds there is a round in which all nodes in level $d-1$ send a message independently with probability $p$ between $\frac{1}{x}$ and $\frac{1}{2x}$ while all nodes with virtual-distance $d$ and $d+1$ (and thus also all other neighbors of $u$) are silent. The probability that $u$ receives a message from any particular neighbor in this round is at least $\frac{1}{2x} (1 - \frac{1}{x})^{x-1} > \frac{1}{8x}$. These events are mutually exclusive and we thus get that the total probability for at least one neighbor successfully transmitting to $u$ during this round is at least $\frac{1}{8}$.
\end{proof}

\smallskip

\begin{proof}[Proof of \Cref{lem:GST-MMV}]
For a large enough constant $\lambda$ let $T=\lambda(D + 2\ceil{\log_2 n} (\log n + \log \frac{1}{\delta}))$. We claim that for any node $v$, the probability that node $v$ does not receive the message in $T$ rounds is at most $\delta$. 

Fix an arbitrary node $v$. To prove the claim, we use \emph{backwards analysis} to view the process of dissemination of the message. In this method, we go back in time, from round $T$ to round $1$, and we find a sequence of collision-free transmissions from source node $s$ to node $v$. Since we are moving back in time, we find this sequence starting from $v$ and going backwards till reaching $s$. 

\smallskip
For each $t$, we say node $u$ is \emph{transmission-connected} to $v$ by backwards time $t$" if there is a sequence of transmissions $u=w_1, w_2, \dots w_\ell=v$ where for each $i\in [1, \ell-1]$, $w_i$ transmits in a round $r_i \in[T-t, T]$, we have $r_{i} < r_{i+1}$, and in round $r_i$, $w_{i+1}$ receives a message from $w_i$. Let $S_t$ be the set of all nodes that are transmission-connected to $v$ by backwards time $t$. Moreover, we then define the potential of $v$ at backwards time $t$ to be $\Phi(t)=\min_{u \in S_t} d_u \ceil{\log_2 n} + l_u$. Note that $\Phi(0) \leq 2\ceil{\log_2 n}^2+ D$. This is because the level of $v$ in $G$ is at most $D$, and the virtual-distance $d_u$ is at most $2 \ceil{\log_2 n}$. To prove the claim, we show that with probability at least $1-2^{-(\log\frac{1}{\delta} + 2 \log n)}$, we have $\Phi(T)=0$. For this, moving backwards in time, we show that in every $8\ceil{\log_2 n}$ interval of consecutive rounds, this potential decreases with probability at least $\frac{1}{16}$ by at least $\ceil{\log_2 n}-1$. For a backwards time $t$, let node $u$ be the node in $S_t$ that minimizes the potential of $v$. The proof is now divided into two cases as follows:

\medskip
\textbf{Case (A)}: Suppose $u$ has at least one $G$-neighbor that has a lower virtual-distance. In this case, \Cref{lem:progslow} guarantees that with probability at least $\frac{1}{8}$ during the rounds in $[T-t - 6 \ceil{\log_2 n}, T-t]$, there is a collision-free transmission from a node $u'$ with $d_{u'} = d_u - 1$ to $u$.
Since $u'$ and $u$ are neighbors their levels $l_u$ and $l_{u'}$ differ at most by one, thus a successful transmission decreases the potential by at least $(d_u \ceil{\log_2 n} + l_u) - (d_{u'} \ceil{\log_2 n} + l_{u'}) = (d_u - d_{u'})\ceil{\log_2 n} - (l_u - l_{u'}) \geq \ceil{\log_2 n} - 1$. Thus, if $u$ has a neighbor with a virtual-distance lower than $d_u$ then with probability at least $\frac{1}{16}$ the potential decreases by at least $\ceil{\log_2 n} - 1$ within any $8 \ceil{\log_2 n}$ rounds when moving backwards in time.

\medskip
\textbf{Case (B)}: Suppose $u$ does not have a $G$-neighbor with a lower virtual-distance. Note that this can only happen if $u=s$ or if there is one directed edge in $G'$ representing a fast stretch, originating from a node $u'$ one level below $u$ in $G'$ and going into $u$. First observe that the starting node of any fast stretch initiates a ``transmission wave'' every $6 \ceil{\log_2 n}$ rounds by creating a new coded packet and sending it as a fast transmission. This packet gets then pipe-lined through the fast stretch with one progress every fast transmission round (that is, once in every two rounds) until it reaches the end of the stretch. Thus, for any node on a fast stretch, there is a new wave arriving every $6\ceil{\log_2 n}$ rounds. Thus, at a time $t' \in [T-t-6 \ceil{\log_2 n}, T-t]$, a fast transmission wave arrives in $u$ and leads to an extended sequence of collision-free transmissions. In particular, if the wave originated from $u'$ during the rounds $[T-t'-2 \ceil{\log_2 n}, T-t']$, then there is a sequence of transmission from $u'$ to $v$ in round interval $[T-t - 8 \ceil{\log_2 n}, T-t]$, and otherwise the wave propagated for $\ceil{\log_2 n}$ steps and there is a node $u''$ between $u'$ and $u$ on the fast stretch with a sequence of transmission to $v$ starting at time $T-t - 8 \ceil{\log_2 n}$. Thus, in both cases, the potential drops by at least $\ceil{\log_2 n}-1$. In the first case the potential drop comes from the fact that $d_{u'} = d_u - 1$ and $l_{u'} < l_u$, while in the second case we have $d_{u''} \leq d_{u'} + 1 = d_u$ and $l_{u''} \leq l_{u} - \ceil{\log_2 n}$. 

\medskip
The above argument shows that when moving backwards in time, in every $8 \ceil{\log_2 n}$ consecutive rounds, with probability at least $\frac{1}{8}$, the potential of $v$ decreases by at least $\ceil{\log_2 n}-1 > \ceil{\log_2 n}/2$, until reaching zero. When the potential reaches zero, it means that there is a sequence of successful and collision-free transmission from $s$ to $v$. 
%
%
Hence, the expected time for such a sequence to appear is thus a constant times the initial potential of $v$, $\Phi_{\vec{\mu}}(0) \leq 2\ceil{\log_2 n}^2 +D$. A Chernoff bound furthermore shows that the probability of not finding such a sequence is exponentially concentrated around this mean. In particular, after $T=\lambda(D + 2\ceil{\log_2 n} (\log n + \log \frac{1}{\delta}))$ rounds, we expect at least $\lambda'(2D/\ceil{\log_2 n} + 4\ceil{\log_2 n} + 2\log{\frac{1}{\delta}})$ sets of $8 \ceil{\log_2 n}$ consecutive rounds in which the potential of $v$ drops at least by $\ceil{\log_2 n}/2$, for a constant $\lambda'$. Furthermore, the probability that there are less than $2D/\ceil{\log_2 n} + 4\ceil{\log_2 n}$ such rounds is exponentially small in the expectation, that is, at most $2^{-(2\ceil{\log_2 n} + \log{\frac{1}{\delta}})} < \delta/n$. A union bound over all choices of node $v$ then completes the proof.
\end{proof}

\subsection{Optimal Multi-Message Broadcast Algorithms}\label{subsec:optMMBA}

We achieve our optimal multi-message broadcast algorithms by combining random linear network coding with the Multi-Message GST Schedule that we presented in \Cref{sec:schedule}. In \Cref{sec:RLNC} we first recall on the exact working of random linear network coding and in \Cref{sec:RLNCintegration} we explain how to integrate it with our MMV GST Schedule. In \Cref{subsec:analysisMMBA} we combine the analysis technique from \cite{Haeupler11} with the proof that our schedule is MMV to obtain \Cref{thm:multipleBcastKnown}, i.e., our multi-message result for the unknown topology setting. In \Cref{subsec:MMB-Unknown} we then discuss how this algorithm can be extended to the unknown topology setting to obtain \Cref{thm:multipleBcastUnknown}.

\subsubsection{Random Linear Network Coding}\label{sec:RLNC}

In random linear network coding~\cite{HKMKE} the $k$ messages are regarded as bit-vectors $\vec {m_1}$, $\ldots \,$, $\vec {m_k} \in \Ftwo^l$ over $\Ftwo$, the finite field of order two. Instead of putting one message in plaintext into a packet nodes transmitt coded packets. Each network coded packet $p$ consists of a linear combination of messages, that is, the vector $\sum_{i=1}^k \alpha_i \vec {m_i} \in \Ftwo^l$. One should think of the coefficient vector $\vec \alpha = (\alpha_1, \ldots, \alpha_k) \in \Ftwo^k$ being transmitted with each message\footnote{In many applications the size of a message is large compared to the $k$ bit coefficient vector which allows sending the coefficient vector with each message with negligible overhead. In our setting increase the packet size to $k$ bits could be too large. Fortunately, the overhead coming from the coefficient vector can be avoided: In the known topology setting there is actually no need for including the coefficient vectors in the packets because using the topology knowledge, all nodes can compute the coefficients offline in a consistent manner. In the unknown topology scenario, using generations, that is, dividing messages into groups of size $\log n$ and then doing network coding only inside each group keeps the coefficient overhead to $O(\log n)$ bits, which is negligible even in our stringent setting (see \Cref{subsec:MMB-Unknown}).}. 

Because of linearity, a node that has a number of these packets can create a packet of this form for any coefficient combination that is spanned by the coefficient vectors of the packets that it has received by that time. Also, if a node has a set of $k$ packets with linearly independent coefficient vectors, then this node can reconstruct all the $k$ messages using Gaussian elimination. In RLNC, every node $u$ stores all its received packets to maintain the subspace that is spanned by them. Whenever $u$ decides to generate a \emph{new coded packet}, it chooses a random coefficient vector from this subspace by taking a random linear combination of the packets stored. Once the subspace spanned by the coefficient vectors in packets received by $u$ is the full space $\Ftwo^k$, then $u$ decodes and reconstructs all the messages.

\subsubsection{Combining the MMV GST Schedule with Random Linear Network Coding}\label{sec:RLNCintegration}

It is now easy to combine random linear network coding with our new GST Schedule:

\medskip\smallskip

\noindent\fbox{
\parbox{0.95\linewidth}{

\smallskip
\textbf{Multi-Message Broadcast Algorithm}: Whenever in MMV schedule of \Cref{sec:schedule}, a node $u$ is prompted to transmit, $u$ transmits a packet determined as follows: (a) if this is a slow transmission, or if this is a fast transmission and $u$ is the first node on a fast stretch, then $u$ transmits a new coded packet, that is, a packet that is created using network coding by combining the messages $u$ has received earlier, (b) if this is a fast transmission but node $u$ is an intermediate node in a fast stretch, then $u$ simply relays the packet it received in the previous fast transmission round (if any). 
}}

\subsubsection{Analyzing the Multi-Message Broadcast Algorithm}\label{subsec:analysisMMBA}

In this section we prove \Cref{thm:multipleBcastKnown} by analyzing the performance of the \emph{multi-message broadcast algorithm} presented in \Cref{subsec:optMMBA}. The analysis combines the proof for the MMV property of the new GST Schedule with the projection analysis from \cite{Haeupler11}. 


\smallskip

The following definition and proposition are taken from \cite{Haeupler11} and form a simple and clean platform for analyzing random linear network coding:


\begin{definition}[{\cite[Definition 4.1]{Haeupler11}}]\label{def:infection}
A node $v$ is \emph{infected} by a coefficient vector $\vec{\mu} \in \Ftwo^k$ if $v$ has received a packet with a coefficient vector $\vec{c} \in \Ftwo^k$ that is not orthogonal to $\mu$, that is, $\left\langle \vec{\mu},\vec{c} \right\rangle \neq 0$.
\end{definition}

\begin{proposition}[{\cite[Lemma 4.2]{Haeupler11}}]\label{prop:NCfacts}
If a node $v$ is infected by a coefficient vector $\vec{\mu}$ and after that, a node $u$ receives a packet from node $v$, then $u$ gets infected by $\vec{\mu}$ with probability at least $1/2$. Furthermore, if a node $v$ is infected by all the $2^k$ coefficient vectors in $\Ftwo^k$, then $v$ can decode all the $k$ messages.
\end{proposition}

With these tools we can proceed to prove \Cref{thm:multipleBcastKnown}:

\begin{proof}[Proof of \Cref{thm:multipleBcastKnown}]
For a large enough constant $\lambda$ let $T=\lambda(D + k \ceil{\log_2 n} + 2\ceil{\log_2 n}^2)$ . We claim that for any node $v$ and any fixed non-zero vector $\vec{\mu} \in  \Ftwo^k$, the probability that node $v$ is not infected by $\vec{\mu}$ in $T$ rounds is at most $2^{-(k + 2 \log n)}$. The proof of this claim is almost identical to the proof of \Cref{lem:GST-MMV}, except that we are want a failure probability $\delta=O(2^{-k})$ and also, we must consider whether each transmission is successful with respect to $\vec{\mu}$ or not. For completeness, we repeat the proof with all details, starting with the next paragraph. Once we have the claim proven, we can conclude via a union bound over all the $2^k$ coefficient vectors in $\Ftwo^k$ that by round  $T$, with high probability, $v$ is infected by all the coefficient vectors in $\Ftwo^k$. That is, by round $T$, $v$ can decode all the $k$ messages. Using another union bound over all the choices of node $v$ then we get that, with high probability, all nodes have received all the messages by round $T$.

Fix a node $v$ and a non-zero vector $\vec{\mu} \in  \Ftwo^k$. To prove the claim, we use \emph{backwards analysis} to view the process of infection spreading of vector $\vec{\mu}$. In this method, we go back in time, from round $T$ to round $1$, and we find a sequence of collision-free transmissions from source node $s$ to node $v$ such that all the transmissions in this chain are successful \emph{with respect to} vector $\vec{\mu}$. Since we are moving back in time, we find this sequence starting from $v$ and going backwards till reaching $s$. 

\smallskip
For each $t$, we say node $u$ is \emph{transmission-connected} to $v$ by backwards time $t$" if there is a sequence of transmissions $u=w_1, w_2, \dots w_\ell=v$ where for each $i\in [1, \ell-1]$, $w_i$ transmits in a round $r_i \in[T-t, T]$, we have $r_{i} < r_{i+1}$, and in round $r_i$, $w_{i+1}$ receives a message from $w_i$. Let $S_t$ be the set of all nodes that are transmission-connected to $v$ by backwards time $t$. Moreover, we then define the potential of $v$ with respect to vector $\vec{\mu}$ at backwards time $t$ to be $\Phi_{\vec{\mu}}(t)=\min_{u \in S_t} d_u \ceil{\log_2 n} + l_u$. Note that $\Phi_{\vec{\mu}}(0) \leq 2\ceil{\log_2 n}^2+ D$.\fullOnly{This is because the level of $v$ in $G$ is at most $D$, and the virtual-distance $d_u$ is at most $2 \ceil{\log_2 n}$.} To prove the claim, we show that with probability at least $1-2^{-(k + 2 \log n)}$, we have $\Phi_{\vec{\mu}}(T)=0$. For this, moving backwards in time, we show that in every $8\ceil{\log_2 n}$ interval of consecutive rounds, this potential decreases with probability at least $\frac{1}{16}$ by at least $\ceil{\log_2 n}-1$. For a backwards time $t$, let node $u$ be the node in $S_t$ that minimizes the potential of $v$. The proof is now divided into two cases as follows:

\medskip
\textbf{Case (A)}: Suppose $u$ has at least one $G$-neighbor that has a lower virtual-distance. In this case, \Cref{lem:progslow} guarantees that with probability at least $\frac{1}{8}$ during the rounds in $[T-t - 6 \ceil{\log_2 n}, T-t]$, there is a collision-free transmission from a node $u'$ with $d_{u'} = d_u - 1$ to $u$, and is successful with respect to $\vec{\mu}$, with probability $1/2$.
Since $u'$ and $u$ are neighbors their levels $l_u$ and $l_{u'}$ differ at most by one, thus a successful transmission decreases the potential by at least $(d_u \ceil{\log_2 n} + l_u) - (d_{u'} \ceil{\log_2 n} + l_{u'}) = (d_u - d_{u'})\ceil{\log_2 n} - (l_u - l_{u'}) \geq \ceil{\log_2 n} - 1$. Thus, if $u$ has a neighbor with a virtual-distance lower than $d_u$ then with probability at least $\frac{1}{16}$ the potential decreases by at least $\ceil{\log_2 n} - 1$ within any $8 \ceil{\log_2 n}$ rounds when moving backwards in time.

\medskip
\textbf{Case (B)}: Suppose $u$ does not have a $G$-neighbor with a lower virtual-distance. Note that this can only happen if $u=s$ or if there is one directed edge in $G'$ representing a fast stretch, originating from a node $u'$ one level below $u$ in $G'$ and going into $u$. First observe that the starting node of any fast stretch initiates a ``transmission wave'' every $6 \ceil{\log_2 n}$ rounds by creating a new coded packet and sending it as a fast transmission. This packet gets then pipe-lined through the fast stretch with one progress every fast transmission round (that is, once in every two rounds) until it reaches the end of the stretch. Thus, for any node on a fast stretch, there is a new wave arriving every $6\ceil{\log_2 n}$ rounds. Moreover, each of these waves is successful with respect to $\vec{\mu}$ with probability at least $1/2$. Thus, at a time $t' \in [T-t-6 \ceil{\log_2 n}, T-t]$, a fast transmission wave arrives in $u$, and with probability $1/2$ leads to an extended sequence of collision-free transmissions that are successful with respect to $\vec{\mu}$. In particular, if the wave originated from $u'$ during the rounds $[T-t'-2 \ceil{\log_2 n}, T-t']$, then there is a sequence of transmission from $u'$ to $v$ in round interval $[T-t - 8 \ceil{\log_2 n}, T-t]$, and otherwise the wave propagated for $\ceil{\log_2 n}$ steps and there is a node $u''$ between $u'$ and $u$ on the fast stretch with a sequence of transmission to $v$ starting at time $T-t - 8 \ceil{\log_2 n}$. Thus, in both cases, the potential drops by at least $\ceil{\log_2 n}-1$. In the first case the potential drop comes from the fact that $d_{u'} = d_u - 1$ and $l_{u'} < l_u$, while in the second case we have $d_{u''} \leq d_{u'} + 1 = d_u$ and $l_{u''} \leq l_{u} - \ceil{\log_2 n}$. 

\medskip
The above argument shows that when moving backwards in time, in every $8 \ceil{\log_2 n}$ consecutive rounds, with probability at least $\frac{1}{16}$, the potential of $v$ decreases by at least $\ceil{\log_2 n}-1 > \ceil{\log_2 n}/2$, until reaching zero. When the potential reaches zero, it means that there is a sequence of successful and collision-free transmission from $s$ to $v$. 
%
%
Hence, the expected time for such a sequence to appear is thus a constant times the initial potential of $v$, $\Phi_{\vec{\mu}}(0) \leq 2\ceil{\log_2 n}^2 +D$. A Chernoff bound furthermore shows that the probability of not finding such a sequence is exponentially concentrated around this mean. In particular, after $T=\lambda (D + k \ceil{\log_2 n}+ 2\ceil{\log_2 n})$ rounds, we expect at least $\lambda'(2D/\ceil{\log_2 n} + 4\ceil{\log_2 n} + k)$ sets of $8 \ceil{\log_2 n}$ consecutive rounds in which the potential of $v$ drops at least by $\ceil{\log_2 n}/2$, for a constant $\lambda'$. Furthermore, the probability that there are less than $(2D/\ceil{\log_2 n} + 4\ceil{\log_2 n}$ such rounds is exponentially small in the expectation, that is, at most $2^{-(2\ceil{\log_2 n} + k)}$. This completes the proof of \Cref{thm:multipleBcastKnown}
\end{proof}

\fullOnly{
\subsection{Extending the Multi-Message Broadcast to the Unkown Topology Setting}\label{subsec:MMB-Unknown}
To achieve \Cref{thm:multipleBcastUnknown}, the key idea is to combine the multi-message broadcast of known topology presented in Sections \ref{sec:schedule} and \ref{subsec:optMMBA} with the idea presented in \Cref{subsec:single-message}, that is, decomposing the graph into rings of width $D'=\frac{D}{\log^4 n}$ layers around the source node using collision detection and then creating one GST for each ring. Here, we present the smaller details that are needed for filling out this outline, to get \Cref{thm:multipleBcastUnknown}.

Recall that our multi-message broadcast algorithm works on top of a GST of graph $G$. In \Cref{sec:GST}, we presented an $O(D\log^4 n)$ distributed GST construction for the unknown topology setting. Refer to \Cref{subsec:GSTdef} for definition of GST and what nodes need to learn in a distributed GST construction. We will use this distributed construction again. However, we first need to enhance it by adding one more element to what nodes learn about GST: In the multi-message broadcast schedule that we presented in \Cref{sec:schedule}, each node $u$ also needs to know the virtual-distance $d_u$ which indicates the directed distance from source $s$ to node $u$ in the virtual graph $G'$ (refer to \Cref{sec:schedule} for definition of $G'$ and the virtual-distance). In the setting with known topology, GST $\mathcal{T}$ and the respective virtual-distance $d_u$ are computed by each node locally without any need for communication between the nodes. In the next lemma, we show that nodes can easily learn these virtual-distances in the unknown topology setting, without changing the asymptotic time complexity of the GST construction. 

\begin{lemma}\label{lem:GSTenhanced} In the radio networks (even without collision detection), there exists a distributed algorithm that, in $O(D \log^4 n)$ rounds, constructs a GST and moreover, each node $u$ also learns its virtual-distance $d_u$ from the source.
\end{lemma}
\begin{proof} First, we construct a GST in $O(D\log^4 n)$ rounds using the construction of \Cref{thm:GSTconst}. We now explain that in $O(D\log^2 n + \log^3 n)$ further rounds, nodes can compute the virtual-distance labels\footnote{Even though faster solutions for this step are possible, since the time complexity will be dominated by that of the GST construction, we only present the slightly less-efficient but cleaner $O(D\log^2 n + \log^3 n)$ solution}. 

Recall from \Cref{lem:virtualDistance} that for each node $u$, we know that $d_u \in [1, 2\ceil{\log n}]$. We compute the virtual-distances in a recursive manner based on the value of $d_u$: Consider a $d \in [1, 2\ceil{2\log n}-1]$ and suppose that all the nodes $u$ that have a distance label $d_u \leq d$ have already learned their distance $d_u$. We explain how to identify the nodes $u$ that have $d_u=d+1$, in $O(D \log n + \log^2 n)$ rounds. 

Let $S_d$ be the set of nodes $u$ that have received virtual-distance label $d_u=d$. Moreover, let $F_d \subseteq S_d$ be the set of nodes in $S_d$ that are the first nodes in a fast stretch. Recall from \Cref{subsec:GSTdef} that since in construction of GST, each node $u$ knows its own rank and the rank of its parent $v$, node $u$ knows whether $u$ is the first node in a fast stretch or its parent $v$ is in the same fast stretch as well. We divide the $O(D \log n+ \log^2 n)$ rounds of recursion of virtual-distance $d+1$ into two stages, with respectively $O(D \log n)$ and $O(\log^2 n)$ rounds, as follows:

In the first stage, we identify all the nodes that are on the fast stretches starting at nodes of $F_d$, and we give all of them virtual-distance label $d+1$. In order to this, we divide this stage between the $\ceil{\log_2 n}$ possible rank values and spend $2D$ rounds on each rank. That is, we first in $2D$ rounds solve the problem for fast stretches of rank $1$ nodes, then in $2D$ rounds solve the problem fast stretches of rank $2$ nodes, etc. For each rank $r\in [1, \ceil{\log^2 n}]$, we spend $2D$ rounds, in two epochs each made of $D$ rounds, as follows: 

The $D$ rounds of the first epoch are as follows: in the $\ell^{th}$ round, each node that is in $F_d$, has rank $r$, and BFS-layer $\ell$ transmits. Each node $u$ that has not received a virtual-distance label before, has BFS-layer $\ell+1$, rank $r$, and receives a message from its parents gets virtual-distance $d_u=d+1$. These $D$ rounds identify the second nodes (those next to the first nodes) in fast stretches of rank $r$, which must receive virtual-distance $d+1$.

The $D$ rounds of the second epoch are as follows: for each $\ell \in [1, D-1]$, if $\ell=1$, then let $S^*$ be the set of nodes that received virtual-distance label $d+1$ in the first epoch, and if $\ell\geq 1$, then let $S^*$ be the set of nodes that received virtual-distance label $d+1$ in the $(\ell-1)^{th}$ round of the second epoch. Then, in the $\ell^{th}$ round, each node $u$ that has not received a virtual-distance label before, has BFS-layer $\ell+1$, rank $r$, and receives a message from its parent gets virtual-distance $d_u=d+1$. 

Note that because of collision-freeness property of GST, all the nodes of fast-stretches of rank $r$ that start in a node in $F_d$ will be identified and will receive distance label $d+1$. After performing the above two epochs for all the ranks $r\in \ceil{\log n}$, we are done with the first stage. Note that the first stage thus takes $D\log n$ rounds, $2D$ rounds for each rank $r\in \ceil{\log n}$.

The second stage is as follows: All nodes in $S_d$ perform $\Theta(\log n)$ phases of the Decay protocol for a total of $\Theta(\log^2 n)$ rounds. Each node $u$ that has not received a virtual-distance label before but receives a message in these rounds sets its virtual-distance label $d_u= d+1$. 
\end{proof}

Now we use this enhanced distributed GST construction to get a multi-message algorithm for the unknown topology with collision detection. 

\begin{proof}[Proof of \Cref{thm:multipleBcastUnknown}] As in the proof of \Cref{thm:singlebcast}, we first use a wave of collisions to get a BFS-layering of the graph. We decompose the graph into rings, each consisting of $D'= \frac{D}{\log^4}$ consecutive BFS-layers, centered around the source node~\footnote{In fact, if $D=O(\log^6)$, then just one ring and thus just one GST  is enough.}. Then, we use the enhanced GST construction presented in \Cref{lem:GSTenhanced} to construct a GST (with the addition of nodes knowing their virtual-distance labels) for each ring, all in time $O(D'\log^4 n) =O(D)$ rounds, by parallelizing the constructions of different rings. 

Suppose that we are done with the construction of the GSTs of the rings. First, let us assume that the coefficient vectors of linear network coding, which consist of at most $k$ bits, fit inside one packet; we later explain how to reduce this overhead to $O(\log n)$.

Let $k' = \frac{D}{\log^3 n}$. Divide the messages into batches, each consisting of at most $k'$ messages. Inside each ring, we can broadcast one batch of messages in $O(D' + k'\log n+ \log^2 n) = O(\frac{D}{\log^4} + \log^2 n)$ rounds, simply using the algorithm of \Cref{subsec:optMMBA} on top of the GST of this ring. To deliver a batch of messages from one ring to another, we simply use forward error correction (FEC)\footnote{Here, FEC can be viewed as a simplified form of network coding as there is no intermediate node in this scenario. That is, the nodes on the outer boundary of one ring transmit and the nodes on the inner boundary of the next ring receive.}. Consider the outer boundary of the $j^{th}$ ring and the inner boundary of the $(j+1)^{th}$ ring, and consider a batch of messages that is already delivered to all nodes in the outer boundary of the $j^{th}$ ring. Then, each of these outer boundary nodes creates $\Theta(k')$ packets using an FEC code such that if a node $w$ receives $\Theta(k')$ of these packets, then $w$ can decode all the $k'$ messages of the batch in consideration. To deliver these FEC coded packets, we use $k'$ phases of the Decay protocol, where the nodes in the outer boundary of the $j^{th}$ ring transmit. It follows from \Cref{lem:decay} and a simple Chernoff bound that after $k' = \Omega(\log n)$ phases of the Decay protocol, each node on the inner boundary of the $(j+1)^{th}$ ring has with high probability received at least $\Theta(k')$ FEC coded packets related to the batch in consideration. Thus, these inner boundary nodes of the $j^{th}$ ring can decode all the messages of this batch. Hence, we conclude that in time $O(D' + k'\log n+ \log^2 n) + O(k'\log n) = O(\frac{D}{\log^4 n} + \log^2 n)$, with high probability, one batches of messages moves from the inner boundary of the $j^{th}$ ring to the inner boundary of the $(j+1)^{th}$ ring. That is, in each $O(\frac{D}{\log^4 n})$, one batch of messages moves one ring forward. 

Having the above, it is enough to pipeline the batches of messages over the rings. That is, the first batch starts in the first ring, and moves one ring forward, in each epoch made of $O(\frac{D}{\log^4} + \log^2 n)$ rounds. When the first batch is in the third ring (and is starting to be broadcasted there), the first ring starts working on the second batch. Note that at each time, nodes in each ring work on at most one batch. This way, the first batch arrives at the end of the last ring by the end of round $O(\frac{D}{\log^4} + \log^2 n) \cdot \log^4 n = O(D + \log^6 n)$. Moreover, after that, in every interval of $O(\frac{D}{\log^4} + \log^2 n)$ consecutive rounds, one new batch arrives at the end of the last ring. Since there are $\frac{k}{k'}$ batches, we get that we are done with the broadcast of all messages by the end of round  $O(D + \log^6 n)+ (\frac{k}{k'}) \cdot O(\frac{D}{\log^4} + \log^2 n) = O(D + \log^6 n)+ (\frac{k\log^3}{D'}) \cdot O(\frac{D}{\log^4} + \log^2 n) = O(D + k\log n + \log^6 n)$.

Lastly, we explain how to reduces the overhead coming from including the coefficient vector into RLNC coded packets from $k$ bits to $O(\log n)$ bits. This is done by grouping all packets into batches of $O(\log n)$ messages and only coding together messages within a batch. This happens only in the transmissions within a ring leaving the process of broadcasting the messages between the boundaries of two consequent rings the same as above, which was fine as the coding overhead of FEC is only a constant. 


Inside each ring, we do the following: Consider the $j^{th}$ ring, for a $j \in [1, \Theta(\log^4)]$, and the GST of that ring. For each node $u$ in this ring, define height of $u$ as $h_u= d_u \ceil{\log_2 n} + l_u$, where $d_u$ is the virtual-distance of $u$ in this ring and $l_u$ is the (normalized) BFS layers of $u$ for this ring (that is, the BFS layer of $u$ in the BFS layering of original graph $G$ minus $j \cdot D'$). Note that this definition of height exactly matches the potential function defined in the proof of \Cref{thm:multipleBcastKnown}. Moreover, note that for each node $u$, we have $h_u \leq 2\ceil{\log n}^2+ D' = O(D'+2\log^2 n)$. Fix $W=\Theta(\log^2 n)$. Based on the height, we decompose the $j^{th}$ ring into \emph{strips} as follows: all nodes $u$ in the $j^{th}$ ring that have $h_u \in [(j'-1) \cdot W, j' \cdot W]$ are in the strip number $j'$.  

Now, to reduce the header overhead caused by coding to $O(\log n)$, instead of dividing the messages into batches of $k' = \frac{D}{\log^3 n}$, we divide them into smaller batches each consisting of $k''=\Theta(\log n)$ messages. Thus, the RLNC coefficient vectors of each batch are $\Theta(\log n)$ bits and hence, fit inside one packet for any packet size $B=\Omega(\log n)$. Now we use the transmission schedule of \Cref{sec:schedule} but with coding the packets only inside one batch and one strip. That is, we run the schedule of \Cref{sec:schedule} in steps consisting of $\Theta(\log^2 n)$ rounds. If a node has not received all the messages of one batch at the end of one step, then it ignores all the packets it received in this step (that is, it empties its buffer) and restarts in the next step. Following the proof of \Cref{thm:multipleBcastKnown}, we see that in each step of $\Theta(\log^2 n)$ rounds, each batch moves one strip forward, with high probability. That is, for each particular batch, in each $\Theta(\log^2 n)$ rounds, the \emph{height} of the nodes that have received all the messages of this batch increases by at least $\Theta(\log^2 n)$, with high probability. Since the maximum height in the ring is $O(D'+2\log^2 n)$, we get that in $O(D'+2\log^2 n)$ rounds, the first batch moves from the start of the ring to the end of the ring. After this, in each $\Theta(\log^2 n)$ further rounds, another batch of messages arrives at the end layer of the ring. From the above, by combining with the pipe-lining argument between different rings, we get that the very first batch reaches the outer boundary of the last ring after $O(D+\log^6 n)$ rounds. After that, in each $\Theta(\log^2 n)$ rounds, one new batch made of $\Theta(\log n)$ messages arrives at the outer boundary of the last ring. Hence, after $O(D+k\log n+ \log^6 n)$ rounds, all batches are broadcast to all nodes of the graph.
\end{proof}

}